%% file: main.tex
\documentclass[a4paper,UKenglish,cleveref, autoref, thm-restate]{lipics-v2021}

\usepackage[utf8]{inputenc}
\usepackage{amsfonts,amssymb,amsmath,amsthm,amscd,latexsym}
\usepackage{xspace} 
\usepackage{mathtools}

\usepackage[pdflatex=true]{gastex}   
\ifpdf
\else
  \AtBeginDocument{%
        \RemoveFromHook{shipout/firstpage}[hyperref]%
        \RemoveFromHook{shipout/before}[hyperref]%
  }
\fi
\usepackage{multirow}
\usepackage{comment}
\usepackage{array}
\usepackage{graphicx}
\usepackage{xcolor}

\definecolor{gray50}{gray}{0.5}
\def\gray50#1{\textcolor{gray50}{#1}}
\definecolor{graydarker}{gray}{0.7}
\def\graydarker#1{\textcolor{graydarker}{#1}}

\definecolor{myred}{rgb}{.7, .12, .12}
\def\myred{\textcolor{myred}}
\definecolor{myblue}{rgb}{.1, .1, .8}
\def\myblue{\textcolor{myblue}}
\definecolor{mygreen}{rgb}{.4, .8, .1}
\def\mygreen{\textcolor{mygreen}}

\makeatletter
\newcommand\footnoteref[1]{\protected@xdef\@thefnmark{\ref{#1}}\@footnotemark}
\makeatother

\def\abs#1{\ensuremath{\lvert #1\rvert}}

\makeatletter
\DeclareRobustCommand\sfrac[1]{\@ifnextchar/{\@sfrac{#1}}%
                                            {\@sfrac{#1}/}}
\def\@sfrac#1/#2{\leavevmode\scalebox{.9}{\kern.1em\raise.5ex
         \hbox{$\m@th\mbox{\fontsize\sf@size\z@
                           \selectfont#1}$}\kern-.1em
         /\kern-.15em\lower.25ex
          \hbox{$\m@th\mbox{\fontsize\sf@size\z@
                            \selectfont#2}$}}}
\DeclareRobustCommand\numfrac[1]{\@ifnextchar/{\@numfrac{#1}}%
                                            {\@numfrac{#1}}}
\def\@numfrac#1{\leavevmode \hbox{$\m@th\mbox{\fontsize\sf@size\z@
                           \selectfont#1}$}}
\makeatother

\newcommand{\nat}{\mathbb N}
\newcommand{\real}{\mathbb R}

\newcommand{\tuple}[1]{\langle #1 \rangle}

\newcommand*{\init}{\varepsilon}

\newcommand{\dist}{{\cal D}}      
\newcommand{\opnu}{\nu}      

\newcommand{\G}{{\cal G}}        
\renewcommand{\H}{{\cal H}}      
\newcommand{\M}{{\cal M}}

\newcommand{\Supp}{{\sf Supp}}
\newcommand{\act}{{\sf act}}

\newcommand{\Reach}{{\sf Reach}}
\newcommand{\Safe}{{\sf Safe}}
\newcommand{\Buchi}{\text{\sf B\"uchi}}
\newcommand{\coBuchi}{\text{\sf coB\"uchi}}

\newcommand{\straa}{\alpha} \newcommand{\Straa}{{\cal A}}
\newcommand{\strab}{\beta}  \newcommand{\Strab}{{\cal B}}

\newcommand{\Prb}{\mathrm{Pr}}
\newcommand{\outcome}{\mathrm{outcome}}
\newcommand{\pure}{{\sf pure}}

\newcommand{\Pre}{{\sf Pre}}

\newcommand{\Cyl}{{\sf Cyl}}

\newcommand{\last}{{\rm last}}

\newcommand{\interior}{{\rm int}}
\newcommand{\U}{\mathcal{U}}

\let\epsilon\varepsilon
\let\emptyset\varnothing

\def\small{\normalfont}

\title{Regular Games with Imperfect Information Are Not That Regular}

\author{Laurent Doyen}{CNRS \& LMF, ENS Paris-Saclay, France}{}{https://orcid.org/0000-0003-3714-6145}{}

\author{Thomas Soullard}{LMF, ENS Paris-Saclay \& CNRS, France}{}{https://orcid.org/0009-0003-2675-6447}{}

\authorrunning{L. Doyen and T. Soullard} 

\Copyright{Laurent Doyen and Thomas Soullard} 



\nolinenumbers 

\begin{document}
\maketitle


\begin{abstract}
We consider two-player games with imperfect information and the 
synthesis of a randomized strategy for one player that ensures the objective
is satisfied almost-surely (i.e., with probability~$1$), regardless of the strategy of the other player.
Imperfect information is modeled by an indistinguishability relation 
describing the pairs of histories that the first player cannot distinguish,
a generalization of the traditional model with partial observations.
The game is regular if it admits a regular function 
whose kernel commutes with the indistinguishability relation.

The synthesis of pure strategies that ensure all possible outcomes satisfy 
the objective is possible in regular games, by a generic reduction that holds
for all objectives. 
While the solution for pure strategies extends to randomized
strategies in the traditional model with partial observations (which is always regular),
a similar reduction does not exist in the more general model.
Despite that, we show that in regular games with B\"uchi objectives
the synthesis problem is decidable for randomized strategies that ensure 
the outcome satisfies the objective almost-surely.
\end{abstract}

\section{Introduction}
We consider the synthesis problem for two-player turn-based games with imperfect information,
a model that has applications in several areas of computer-science, including 
the design of multi-agent systems~\cite{GGL22}, 
logics with uncertainty in planning and AI~\cite{BBvDM17,DF23},
program synthesis \cite{CCHRS11}, and automata theory~\cite{ABK22}.
The synthesis problem asks to decide the existence of (and if so, to construct)
a strategy for one player that ensures a given objective is satisfied with the largest possible probability,
regardless of the choices of the other player.
We focus mostly on reachability objectives, and will also discuss $\omega$-regular objectives~\cite{automata}.

Synthesis provides a natural formulation for the design of a reactive system
that interacts with an unknown environment. The interactive nature of reactive
systems is modeled by a two-player game between the system (the player) and an adversarial environment,
and the limited access of the system to the current state of the game is modeled
by imperfect information.

A simple example of a turn-based game with imperfect information is repeated matching pennies, where
the environment (secretly) chooses head or tail (say, denoted by $x \in \{0,1\}$), 
then the player chooses $y \in \{0,1\}$ without seeing $x$, and wins if $y=x$.
If $y \neq x$, the game repeats for one more round; the player loses if the game repeats forever. 
It is clear that there exists no pure (deterministic) strategy that is winning in repeated
matching pennies, as for all sequences $\bar{x} = x_1 \ldots x_i$
and $\bar{y} = y_1 \ldots y_i$ that represent a history of the game (namely, the sequence of moves played in 
the first $i$ rounds), given $x_{i+1}$ (chosen by the environment), the choice of $y_{i+1}$ (by the player)
must be made independently of $x_{i+1}$, thus we may have $x_{i+1} = 1-y_{i+1}$, showing
that the game may repeat forever. That the player cannot ensure to win is counter-intuitive,
and this is because pure (deterministic) strategies are considered. With randomized strategies, 
the player can win with probability~$1$, simply by choosing $y_i \in \{0,1\}$
with uniform probability at every round. The power of randomization for decision making
with imperfect information is well known and occurs in many situations~\cite[Section~1.3]{Gibbons92}, 
including gene mutation in biology, penalty kick in sports, bluffing in card games, etc.
It is also known that, even for simple objectives like reachability, optimal strategies
may need memory~\cite{CDHR07}. 

In the sequel, we distinguish the synthesis of a pure strategy (which we call
a sure winning strategy, as it aims to ensure all possible outcomes satisfy the objective),
and the synthesis of a randomized strategy (which we call an almost-sure 
winning strategy, as it aims to ensure that the objective is satisfied with probability~$1$).
Note that the synthesis problem for strategies that 
ensure the objective is satisfied with probability at least $\lambda$ (given $0 < \lambda < 1$)
is undecidable for probabilistic automata~\cite{Paz-Book}, a model that can be 
reduced to two-player turn-based games with imperfect information~\cite{CDGH15}.

Given a finite alphabet $\Gamma$ of moves, we denote by $\Gamma^*$ the set of
all histories in a game.
The traditional model of imperfect information in games is a (regular) observation function
defined on $\Gamma^*$ that assigns to each history of the game a color that is visible to the player,
while the history itself is not visible~\cite{Rei84}. 
We call them partial-observation games, or games {\it \`a la Reif}.
This model of games with imperfect information (even turn-based and non-stochastic)
with randomized strategies generalizes many models such as concurrent games (e.g., 
matching pennies), Markov chains, Markov decision processes, 
and stochastic games~\cite[Theorem~5]{CDGH15}. 
The synthesis of pure strategies in partial-observation games {\it \`a la Reif} 
with $\omega$-regular objectives can be done using automata-based techniques 
analogous to the subset construction for finite automata~\cite{Rei84,CDHR07}.
The synthesis of randomized (almost-sure winning) strategies can be done
for reachability (and B\"uchi) objectives with a more involved technique
based on the same approach of subset construction. 
The synthesis problem is EXPTIME-complete in both cases.
Note that for $\omega$-regular (and even for coB\"uchi) objectives,
the synthesis problem with randomized strategies is undecidable~\cite{BGB12}.  

A more abstract (and more general) model of imperfect information is given by a function 
$f: \Gamma^* \times \Gamma^* \to \{0,1\}$ that specifies an indistinguishability
relation~\cite{BD20,BDS23}: two histories $\tau,\tau'$ are indistinguishable if $f(\tau,\tau')=1$.
In analogy to the partial-observation model, indistinguishable histories
have the same length and their prefixes (of the same length) are also indistinguishable.
Hence this function can be viewed as a set of pairs of finite words of the same 
length (over alphabet $\Gamma$), or alternatively as a language of finite words 
over the alphabet $\Gamma \times \Gamma$. Applications of this model include
the synthesis of full-information protocols, which
cannot be expressed by a partial-observation game {\it \`a la Reif}~\cite[Section~6]{BDS23}.

A generic approach to the synthesis problem with this more general model
is to construct a so-called rectangular morphism $h$ defined on $\Gamma^*$ 
with finite range~\cite{BDS23}. Note that a morphism with finite range is a regular function
(which can be defined by a finite-state automaton with output), and we say
that a game is \emph{regular} if such a rectangular morphism $h$ exists. 
For instance, all partial-observation games {\it \`a la Reif} are regular 
as they admit a natural rectangular morphism, which is essentially the subset construction.
The existence of a rectangular morphism guarantees that the information
tree has a finite bisimulation quotient, and thus (pure) strategies 
can be transferred back and forth (by copying the action played in a bisimilar 
position of the other game), preserving the outcome of the game for arbitrary objectives~\cite{BDS23}.
Hence, for the synthesis of pure strategies, regular games with imperfect information 
can be reduced to finite games with perfect information.

In this paper, we consider the synthesis problem with randomized strategies, 
which is central in games with imperfect information (such as matching pennies).
Given the existence of a finite bisimulation quotient in regular games,
the fact that a solution based on a rectangular morphism (the subset construction)
works for the synthesis of randomized strategies in reachability games {\it \`a la Reif}
gives hope that synthesis in regular games is solvable
along a similar path, by a reduction to a simpler equivalent game from which randomized 
strategies can be transferred back and forth, copying both the action
and the probability of playing that action. The hope is reinforced
by the existence of a bijection between the set of randomized strategies in a
game {\it \`a la Reif} and the set of strategies in the simpler equivalent game,
which relates strategies with same (probabilistic measure on the) outcome
for arbitrary objectives~\cite[Theorem~4.2]{CDHR07}.

Surprisingly, we show that even a much weaker variant of such a reduction 
does not exist for regular games. First the hope for a bijective transfer
of strategies is not reasonable because the information tree of
regular games have in general unbounded branching~\cite{BD20}, whereas the simpler 
equivalent games induced by (finite) rectangular morphism have a
bounded-branching information tree. On the other hand, we show that even 
a non-bijective transfer of strategies, as general as it can reasonably be, may not exist
for some regular game with a specific rectangular morphism (Section~\ref{sec:cex}). 
Despite their nice structural properties, the existence of a finite bisimulation
in their information tree, and their apparent similarity with games {\it \`a la Reif}
where a bijection between randomized strategies is induced by a rectangular morphism,
regular games are not as well-behaved as games {\it \`a la Reif}. 

In this context, we 
present an algorithmic solution of the 
synthesis problem with randomized strategies for regular games 
with a reachability objective. The solution exploits the
properties of rectangular morphisms to define a fixpoint computation
with complexity quadratic in the size of the range of the rectangular morphism
(which is of exponential size in the case of games {\it \`a la Reif}~\cite{CDHR07},
and of non-elementary size in the case of full-information protocols~\cite{BDS23}).
Our algorithm shares the common features of the solutions of almost-sure reachability
objectives in Markov decision processes~\cite{CY95}, concurrent games~\cite{AHK07}, 
and games {\it \`a la Reif}~\cite{CDHR07}, namely a nested fixpoint computation
that iteratively constructs the almost-sure winning set 
by computing the set $S_{i+1}$ of states from which the player can win with 
positive probability while ensuring to never leave the set $S_i$, where $S_0$
is the entire state space (Section~\ref{sec:as-reachability}). 
Our solution immediately extends to B\"uchi objectives, using reductions
to reachability objectives from the literature. 

In conclusion, this paper generalizes the positive (decidability) results about randomized strategies
from games {\it \`a la Reif} to regular games (namely for reachability and B\"uchi
objectives), whereas the reductions that worked in games {\it \`a la Reif}
for arbitrary objectives no longer hold.

Omitted proofs are provided in the appendix.


\section{Definitions}

We recall basic definitions from logic and automata theory that will be useful
in the rest of the paper, and then we discuss our model of games.

\subsection{Basic notions}

\smallskip\noindent{\em Regular functions.} A \emph{Moore machine} is a finite-state automaton with outputs, 
consisting of a finite input alphabet $\Gamma$, a finite output alphabet $\Sigma$,
and a tuple $\M = \tuple{Q, q_\init, \delta, \lambda}$ 
with a finite set $Q$ of states (sometimes called the memory), an initial state $q_\init \in Q$,
a (deterministic) transition function $\delta: Q \times \Gamma \to Q$,
and an output function $\lambda: Q \to \Sigma$.

We extend the transition function to input words in the natural way,
defining $\delta: Q \times \Gamma^* \to Q$ by $\delta(q,\epsilon) = q$
for all states $q \in Q$, and inductively $\delta(q,\tau c) = \delta(\delta(q,\tau), c)$
for all histories $\tau \in \Gamma^*$ and moves $c \in \Gamma$. 

The object of interest is the extension of the output function to
a function $\lambda: \Gamma^* \to \Sigma$ defined by 
$\lambda (\tau) = \lambda(\delta(q_\init,\tau))$ for all histories $\tau \in \Gamma^*$.
A function $\lambda$ defined on $\Gamma^*$ is \emph{regular} if it is (the extension of)
the output function of a Moore machine.
The \emph{cumulative extension} of $\lambda$ is denoted by $\hat{\lambda}: \Gamma^* \to \Sigma^*$ 
and defined by $\hat{\lambda}(\epsilon) = \epsilon$
and inductively $\hat{\lambda}(\tau c) = \hat{\lambda}(\tau) \lambda(\tau c)$,
thus $\hat{\lambda}(c_1\ldots c_k) = \lambda(c_1) \lambda(c_1c_2) \ldots \lambda(c_1\ldots c_k)$.
Note that $\lambda = \last \circ \hat{\lambda}$ where $\last(c_1\ldots c_k) = c_k$.
The function $\hat{\lambda}$ is further extended to infinite words as expected:
$\hat{\lambda}(\pi) = \lambda(c_1) \lambda(c_1c_2) \ldots$ for all $\pi = c_1c_2\ldots \in \Gamma^{\omega}$.

\smallskip\noindent{\em Fixpoint formulas.}
We briefly recall the interpretation of  $\mu$-calculus fixpoint formulas~\cite{BS07,BW18}
based on the Knaster-Tarski theorem.
Given a monotonic function $\psi: 2^Q \to 2^Q$
(i.e., such that $X \subseteq Y$ implies $\psi(X) \subseteq \psi(Y)$),
the expression $\opnu Y. \psi(Y)$ is the (unique) greatest fixpoint of $\psi$,
which can be computed as the limit of the sequence $(Y_i)_{i \in \nat}$ defined by
$Y_0 = Q$, and $Y_{i} = \psi(Y_{i-1})$ for all $i \geq 1$.
Dually, the expression $\mu X.~\psi(X)$ is the (unique) least fixpoint of $\psi$,
and the limit of the sequence $(X_i)_{i \in \nat}$ defined by
$X_0 = \emptyset$, and $X_{i} = \psi(X_{i-1})$ for all $i \geq 1$.
If $\abs{Q} = n$, then it is not difficult to see that the 
limit of those sequences is reached after at most $n$ iterations, 
$X_{n} = X_{n+1}$ and $Y_{n} = Y_{n+1}$.

\smallskip\noindent{\em On equivalence relations.}
Given an equivalence relation $\sim\, \subseteq S \times S$
and a set $T \subseteq S$, the \emph{closure} of $T$ under $\sim$ is the set 
$[T]_{\sim} = \{t \in S \mid \exists t' \in T: t \sim t' \}$,
and the \emph{interior} of $T$ under $\sim$ is the set 
$\interior_{\sim}(T) = \{t \in S \mid [t]_{\sim} \subseteq T\}$.
The closure and interior are monotone operators, that is if
$T \subseteq U$, then $[T]_{\sim} \subseteq [U]_{\sim}$ and 
$\interior_{\sim}(T) \subseteq \interior_{\sim}(U)$.
Note the duality $S \setminus \interior_{\sim}(T) = [S \setminus T]_{\sim}$,
the complement of the interior of $T$ is the closure of the complement of $T$.
We say that $T$ is \emph{closed} under $\sim$ if $T = [T]_{\sim}$ (or equivalently
$T = \interior_{\sim}(T)$).

\subsection{Games with imperfect information}\label{sec:def:games}

We consider an abstract model of games, given by a (nonempty) set $A$
of actions, a set $\Gamma$ of moves, and a surjective function $\act: \Gamma \to A$. 
The game is played over infinitely many rounds (or steps). 
In each round, the player chooses an action $a \in A$, 
then the environment chooses a move $c \in \Gamma$ such that $\act(c) = a$.
We say that the move $c$ is \emph{supported} by the action $a$.

The outcome of the game is an infinite sequence $\pi \in \Gamma^{\omega}$ of 
moves of the environment, called a \emph{play},
from which we can reconstruct the infinite sequence of actions of the player
using the function $\act$. A finite sequence of moves is called a \emph{history}.

The winning condition is a set $W \subseteq \Gamma^{\omega}$ of plays
that are winning for the player. We specify the winning condition
by the combination of a coloring function $\lambda: \Gamma^* \to C$ 
for some finite alphabet $C$ of colors, and a logical objective $L \subseteq C^{\omega}$.
The coloring function is regular, and specified by a Moore machine.
The induced winning condition is the set 
$W = \{\pi \in \Gamma^{\omega} \mid \hat{\lambda}(\pi) \in L\}$
of plays whose coloring satisfies the objective.

We consider the following classical objectives, for $C = \{0,1\}$:
\begin{itemize}
\item the reachability objective is $\Reach = 0^*1\{0,1\}^{\omega}$
and for $x \in \nat$ let $\Reach^{\leq x} = \bigcup_{i \leq x} 0^i 1\{0,1\}^{\omega}$;

\item the safety objective is $\Safe = 0^{\omega} = \{0,1\}^{\omega} \setminus \Reach$;

\item the B\"uchi objective is $\Buchi = (0^*1)^{\omega}$;

\item the coB\"uchi objective is $\coBuchi = \{0,1\}^*0^{\omega} = \{0,1\}^{\omega} \setminus \Buchi$;

\end{itemize}

It is convenient to consider that the coloring function $\lambda$ is part of 
a game instance, while the objective $L$ can be specified independently of 
the game. This allows to quantify separately the game instances
and the objectives. We recall that the traditional games played on graphs~\cite{automata,AG11} 
are an equivalent model, in particular they can be translated into our abstract
model of game, for example by letting moves be the edges of the game graph and the 
corresponding action be the label of the edge, while the graph structure
can be encoded in the winning condition~\cite{Thomas95}. When we refer
to a game played on a graph, we assume that all move sequences that contain 
two consecutive moves (i.e., edges) $c_i = (q_1,q_2)$ and $c_{i+1} = (q_3,q_4)$ such 
that $q_2 \neq q_3$ belong to the winning condition (for the player).

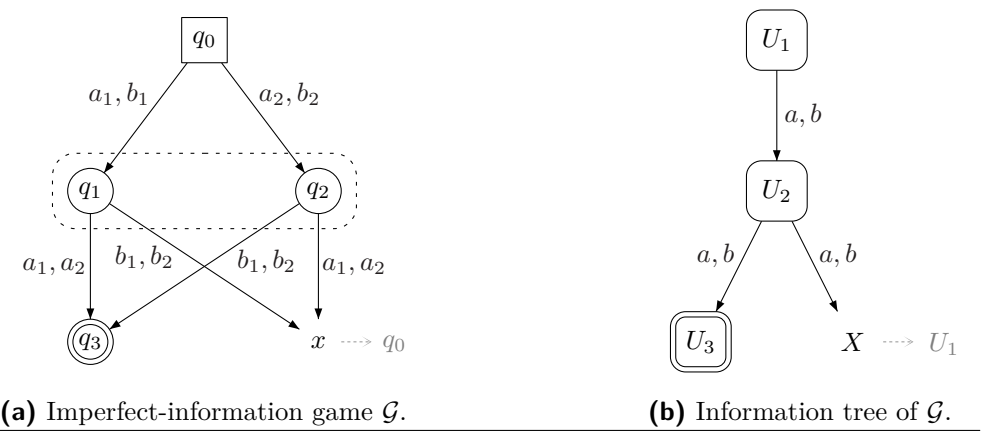
\begin{figure}[!t]%
\begin{center}
\hrule
\medskip
\hspace{10mm}
\subfloat[Imperfect-information game $\G$.]{
   \input{figures/game-as.tex}
   \label{fig:game-as}  
}
\hfill
\subfloat[Information tree of $\G$.]{
   \input{figures/game-as-information-tree.tex}
    \label{fig:game-as-information-tree}
}
\hfill
\hrule
\caption{A game with imperfect information representing matching pennies, and its information tree.
 \label{fig:game-as-main}}%
\end{center}
\end{figure}

Imperfect information is represented by an indistinguishability relation
$\sim \, \in \Gamma^* \times \, \Gamma^*$ that specifies the pairs of histories 
that the player cannot distinguish. 
An indistinguishability relation is an equivalence that satisfies the 
following conditions~\cite{BD20}, 
for all histories $\tau, \tau' \in \Gamma^*$ and moves $c, c' \in \Gamma$:
$(i)$ if $\tau \sim \tau'$, 
then $\abs{\tau} = \abs{\tau'}$ (indistinguishable histories have the same length); 
$(ii)$ if $\tau c \sim \tau' c'$, 
then $\tau \sim \tau'$ (the relation is prefix-closed) and $\act(c) = \act(c')$ (the action is visible).
Indistinguishability relations may be specified using two-tape automata, 
but the results of this paper hold for arbitrary indistinguishability relations.

Given a history $\tau \in \Gamma^*$
we call the equivalence class $[\tau]_{\sim} = \{\tau' \in \Gamma^* \mid \tau' \sim \tau\}$ 
containing~$\tau$ an \emph{information set}.
A function $\lambda$ defined on $\Gamma^*$ is \emph{information consistent} 
if $\lambda$ is 
constant over every information set: 
$\lambda(\tau) = \lambda(\tau')$ for all indistinguishable histories $\tau \sim \tau'$. 
An equivalent requirement is that the cumulative extension $\hat{\lambda}$
is constant over every information set, since indistinguishability relations are
prefix-closed.
We always require the coloring function of a game to be information consistent
(with respect to the indistinguishability relation of the game).
For every action $a \in A$, we define a successor relation $\to_a$ over information sets,
where $[\tau]_{\sim}  \to_a [\tau c]_{\sim}$ for all moves $c \in \Gamma$ such that $\act(c) = a$.
The relations $\to_a$ are obviously acyclic, and induce a structure called 
the \emph{information tree}. Note that there is
a natural extension of the coloring function~$\lambda$ to information sets, defined by
$\lambda([\tau]_{\sim}) = \lambda(\tau)$ for all histories $\tau \in \Gamma^*$,
since $\lambda$ is information-consistent.

Given the above definitions, we represent a game with imperfect information 
as a tuple $\G = \tuple{A,\act,\sim,\lambda}$, together with an 
objective $L \subseteq C^{\omega}$, where the (finite) alphabet of moves is
$\Gamma$, and the (finite) alphabet of colors is $C$.
The information tree of $\G$ is denoted by $\U_{\sim}(\G)$ (consisting
of the information sets of $\G$, the successor relations $\to_a$, and the coloring $\lambda$
extended to information sets).
The special case of perfect-information games corresponds 
the indistinguishability relation $\sim$ being the identity (or equivalently 
by the information sets $[\tau]_{\sim} = \{\tau\}$ being singletons 
for all histories $\tau \in \Gamma^*$).

In figures, we present a very informal (but readable) description of games.
The main components of a game that figures represent are the
coloring function and the indistinguishability relation. 
Strictly speaking, the figures do not show a game but provide sufficient
information to reconstruct the Moore machine defining the coloring function,
and the two-tape automaton defining the indistinguishability relation.
It should be possible to infer the color of a history (when its color is relevant),
and the information set containing a given history.
States are normally depicted as circles, but we may use boxes to emphasize
that the action choice of the player is not relevant and only the move of the
environment determines the successor state.

In the game $\G$ of \figurename~\ref{fig:game-as}, representing matching pennies,
the player has two actions, $A = \{a,b\}$, and the environment has two possible responses
for each of them, thus four moves, $\Gamma = \{a_1,a_2,b_1,b_2\}$
where $\act(a_i) = a$ and $\act(b_i) = b$ for $i=1,2$.
Intuitively, the adversarial environment chooses head ($1$) or tail ($2$) 
in the first round (at $q_0$) by choosing a response $1$ or $2$ (the action of the player is not relevant) 
that the player cannot observe (all histories of length $1$ are indistinguishable,
as suggested by the dashed line in \figurename~\ref{fig:game-as}).
The player chooses head (action $a$) or tail (action $b$) in the next round, 
wins if the choice is matching the environment's choice,
and replays the same game otherwise ($x$ is a placeholder for the whole tree of 
\figurename~\ref{fig:game-as}).

The information tree of $\G$ is shown in \figurename~\ref{fig:game-as-information-tree}
where $U_1 = [\epsilon]_{\sim}  = \{\epsilon\}$ and $U_2 = [a_1]_{\sim}  = \{a_1,a_2,b_1,b_2\} = \Gamma$.
The information set $U_3$ contains eight of the sixteen histories of length $2$.
The node $X$ is a placeholder for a tree that is isomorphic to the tree rooted at $U_1$.

We recall that the nodes in information trees have a finite number of successors 
(there is at most $\abs{\Gamma}^n$ information sets containing histories of length $n$),
however the branching degree in a given information tree may be unbounded~\cite{BD20}.
The partial-observation games~{\it \`a la Reif}~\cite{Rei84} are characterized
by their information tree having bounded branching~\cite[Theorem~7]{BD20}.


\subsection{Strategies and outcome}\label{sec:def:strat}
A \emph{probability distribution} on a finite set~$S$ is a
function $d : S \to [0, 1]$ such that $\sum_{s \in S} d(s)= 1$. 
A \emph{Dirac distribution} assigns probability~$1$ to a (unique) element $s \in S$.
We denote by $\dist(S)$ the set of all probability distributions on~$S$.

A strategy for the player is a function $\straa: \Gamma^* \to \dist(A)$
that is information consistent, thus $\straa(\tau) = \straa(\tau')$ for all 
indistinguishable histories $\tau \sim \tau'$. 
The strategy $\straa$ is \emph{pure} (or, deterministic) if $\straa(\tau)$
is a Dirac distribution for all $\tau \in \Gamma^*$.
A strategy for the environment is a function 
$\strab: \Gamma^* \times A \to \dist(\Gamma)$
such that for all histories $\tau \in \Gamma^*$,
actions $a \in A$, and moves $c \in \Gamma$, 
if the move $c$ is played with positive probability $\strab(\tau,a)(c) > 0$,
then it is supported by the action $a$, $\act(c) = a$.
The strategy $\strab$ is \emph{pure} (or, deterministic) if $\strab(\tau,a)$
is a Dirac distribution for all $\tau \in \Gamma^*$ and $a \in A$.
We denote by $\Straa_{\G}$ (resp., by $\Strab_{\G}$) the set of all strategies 
for the player (resp., for the environment), and by $\Straa^{\pure}_{\G}$ (resp., by $\Strab^{\pure}_{\G}$)
the set of all pure strategies for the player (resp., for the environment).

A pair of strategies $(\straa,\strab)$ for the player and the environment
induce a probability measure $\Prb^{\straa,\strab}$ on the sigma-algebra
over the set of (infinite) plays. By Carath\'eodory's extension theorem,
it is sufficient to define the probability measure over cylinder sets spanned
by (finite) prefixes of plays. We define a family of 
probability measures $\Prb_{\tau}^{\straa,\strab}$ where $\tau \in \Gamma^*$
corresponds to starting the game after history $\tau$.
We set $\Prb^{\straa,\strab} = \Prb_{\epsilon}^{\straa,\strab}$.
Given $\rho \in \Gamma^*$, the measure of the cylinder set
$\Cyl(\rho) = \rho \Gamma^{\omega} = \{\pi \in \Gamma^{\omega} \mid 
\rho \text{ is a prefix of } \pi \}$ is defined as follows:

\begin{itemize}

\item if $\rho$ is a prefix of $\tau$, then 
$$\Prb_{\tau}^{\straa,\strab}(\Cyl(\rho)) = 1;$$

\item if $\rho = \tau c_1 \dots c_k$ where $c_1, \dots, c_k \in \Gamma$
for some $k \geq 1$,
then 
$$\Prb_{\tau}^{\straa,\strab}(\Cyl(\rho)) = 
\prod_{i=1}^{k} \straa(\tau c_1 \dots c_{i-1})(\act(c_i)) \cdot \strab(\tau c_1 \dots c_{i-1}, \act(c_i))(c_i);$$

\item all other cylinder sets have measure~$0$.

\end{itemize}

Given an objective $L \subseteq C^{\omega}$ defined on colors, 
assuming the game $\G$ (and thus $\hat{\lambda}$) is clear from the context
we take the freedom to denote by $\Prb^{\straa,\strab}(L)$
the probability $\Prb^{\straa,\strab}(\{\pi \in \Gamma^{\omega}
\mid \hat{\lambda}(\pi) \in L\})$ that the objective $L$ is satisfied
by an outcome of the pair of strategies $(\straa,\strab)$ in $\G$.
A strategy $\straa$ of the player is \emph{almost-sure winning} from an information
set $[\tau]_{\sim}$ (or simply \emph{almost-sure winning} if $\tau = \epsilon$) 
for objective $L$ if $\Prb_{\tau'}^{\alpha,\beta}(L) = 1$ 
for all $\tau' \in [\tau]_{\sim}$ and all strategies $\beta \in \Strab_{\G}$ of the environment.
We are interested in solving the synthesis problem, which is to decide,
given a game $\G$ and objective $L$, whether there exists an almost-sure winning 
strategy in $\G$ for $L$.

Note that the definition of almost-sure winning is equivalent to a formulation
where only \emph{pure} strategies of the environment are considered, 
that is where we require $\Prb_{\tau'}^{\alpha,\beta}(L) = 1$ for all $\beta \in \Strab^{\pure}_{\G}$ only, 
as once a strategy for the player is fixed, the environment plays in a (possibly infinite-state) 
Markov decision process, for which pure strategies are sufficient~\cite{Mar98,CDGH15}.


We say that a history $\tau \in \Gamma^*$ is \emph{compatible} with $\straa$ and $\strab$
if $\Prb^{\straa,\strab}(\Cyl(\tau)) > 0$. An \emph{outcome} of the pair of strategies $(\straa,\strab)$
is a play $\pi \in \Gamma^{\omega}$ all of whose prefixes are compatible with $\straa$ and $\strab$.
We note that if $\Prb^{\alpha,\beta}(\Cyl(\tau)) > 0$,
then $\Prb_{\tau}^{\alpha,\beta}(\Cyl(\tau c_1 c_2 \dots c_n)) =
\Prb^{\alpha,\beta}(\Cyl(\tau c_1 c_2 \dots c_n) \mid \Cyl(\tau))$.

\section{Almost-Sure Reachability}\label{sec:as-reachability}

Throughout this section, we consider games with a (fixed) reachability objective $\Reach = 0^*1\{0,1\}^{\omega}$
over color alphabet $C = \{0,1\}$,
thus a play is winning if it has a finite prefix with color $1$ according to the coloring function $\lambda$
of the game. 
By extension, we say that a history $\tau$ is winning if $\lambda(\tau) = 1$.
We also assume without loss of generality 
that if $\lambda(\tau) = 1$, then $\lambda(\tau c) = 1$ for all
$\tau \in \Gamma^*$ and all $c \in \Gamma$. 
In the rest of this section, we fix a game $\G$ as defined in Section~\ref{sec:def:games}
and we show how to decide the existence
of an almost-sure winning strategy in games with a reachability 
objective, relying on the existence of a rectangular morphism (defined below).
Two interesting classes of games are known to admit rectangular morphisms,
namely the partial-observation games~{\it \`a la Reif} and the full-information
protocols~\cite{BDS23}.

\subsection{Regular games with imperfect information}\label{sec:regular-games}

A \emph{rectangular morphism} for $\G$ is 
a surjective function $h: \Gamma^* \to P$ (for some finite set $P$)
such that for all $\tau, \tau' \in \Gamma^*$ and $c \in \Gamma$:
\begin{description}

\item[Rectangularity\footnotemark]\footnotetext{Rectangularity is equivalent to the
kernel $H = \{(\tau,\tau') \mid h(\tau) = h(\tau')\}$ of $h$ commuting 
with the indistinguishability relation, that is $H \,\circ \sim \;=\; \sim \circ\, H$.}
if $h(\tau) = h(\tau')$, 
then $h([\tau]_{\sim}) = h([\tau']_{\sim})$,

\item[Morphism] if $h(\tau) = h(\tau')$, 
then $h(\tau c) = h(\tau' c)$,

\item[Refinement] if $h(\tau) = h(\tau')$, 
then $\lambda(\tau) = \lambda(\tau')$.
\end{description}

Note that, as $h$ is a finite-valued morphism, it can be represented by an
automaton (whose output is not relevant) with input alphabet $\Gamma$ 
and state space $P$. Games that admit a finite-valued rectangular morphism are called \emph{regular}.
Define the relation $\approx\, = \{(h(\tau), h(\tau')) \mid \tau \sim \tau'\}$
and recall a non-trivial fundamental result.

\begin{lemma}\cite[Lemma~3 \& Lemma~4]{BDS23}\label{lem:approx-equiv}
The relation $\approx$ is an equivalence and $h([\tau]_{\sim}) = [h(\tau)]_{\approx}$
for all $\tau \in \Gamma^*$.
\end{lemma}

We call the elements of $P$ \emph{abstract states}. 
We extend the relation $\approx$ to elements of $P \times A$
as follows: for all $p,p' \in P$ and $a,a' \in A$, let 
$(p,a) \approx (p',a')$ if $p \approx p'$ and $a=a'$.

Thanks to the morphism property, there exists a function $\delta^P: P \times \Gamma \to P$ such that 
$\delta^P(h(\tau),c) = h(\tau c)$ for all $\tau \in \Gamma^*$ and $c \in \Gamma$.
Define the set 
$P_F = \{p \in P \mid \exists \tau \in \Gamma^*: h(\tau) = p \land \lambda(\tau) = 1 \}$ 
of target (or final) states, which are the images by $h$ of winning
histories. Note that the coloring function is information-consistent
and therefore the set $P_F$ is closed under $\approx$ by Lemma~\ref{lem:approx-equiv}.
Moreover, by our assumption
on the coloring function the states in $P_F$ form a sink set, that is $\delta^P(p,c) \in P_F$
for all states $p \in P_F$ and moves $c \in \Gamma$.

Rectangular morphisms are central to the solution of the synthesis problem for \emph{pure}
strategies~\cite{BDS23}, by showing that the following abstract game is equivalent
to $\G$: starting from $p_0 = h(\epsilon)$, the game is played in rounds where
each round starts with a value $p$ and the player chooses an action $a \in A$,
the next round starts in $p'$ chosen by the environment such that $p' \approx \delta^P(p,c)$
for some move $c \in \Gamma$ such that $\act(c) = a$. The player wins if a value in $P_F$
eventually occurs along a play. 
We may view the elements of $P$ as positions of the player, and elements
of $P \times A$ as positions of the environment. We further discuss this 
abstract game in Section~\ref{sec:bij-bis}.


An abstract state $p \in P$ is \emph{existentially winning}
if there exists a history $\tau \in \Gamma^*$
such that $h(\tau) = p$ and 
the player has an \emph{almost-sure winning} strategy from $[\tau]_{\sim}$.
We denote by $P_{\exists}$ the set of all existentially winning states.
Note that the set $P_{\exists}$ is closed under $\approx$,
since for all $p \in P_{\exists}$ and $p'\approx p$, by Lemma~\ref{lem:approx-equiv}
there exists $\tau'\in [\tau]_{\sim}$, such that $h(\tau')=p'$
and the player has an \emph{almost-sure winning} strategy from $[\tau']_{\sim} = [\tau]_{\sim}$.

An abstract state $p \in P$ is \emph{universally winning}
if for all histories $\tau \in \Gamma^*$ such that $h(\tau) = p$,
the player has an \emph{almost-sure winning} strategy from $[\tau]_{\sim}$.
We denote by $P_{\forall}$ the set of all universally winning states.
Note that the set $P_{\forall}$ is closed under $\approx$, by an argument
analogous to the case of $P_{\exists}$.
As we consider a reachability objective, it is easy to see that 
$P_F \subseteq P_{\forall} \subseteq P_{\exists}$ (the latter
inclusion holds since $h$ is surjective). 

Note that in games of perfect information, the existentially and universally winning
states coincide by definition since information sets are singletons, $P_{\forall} = P_{\exists}$.
A corollary of our results is that $P_{\forall} = P_{\exists}$ also 
holds in games of imperfect information with a reachability objective.

\subsection{Algorithm}

We present an algorithm to decide if the player is almost-sure winning
for reachability in a game of imperfect information, assuming we have a
rectangular morphism.  

Given a set $X \subseteq P \cup (P \times A)$, define the 
\emph{predecessor} operator as follows:

\begin{align*}
\Pre(X) = \; & \{p \in P \mid \exists a \in A: (p,a) \in X\} \\
             & \cup \\
             & \{(p,a) \in P \times A \mid \forall c : \act(c) = a \implies \delta^P(p,c) \in X\}
\end{align*}

It is easy to verify that if
$X_1 \subseteq X_2$, then $\Pre(X_1) \subseteq \Pre(X_2)$,
that is $\Pre(\cdot)$ is monotone.
Intuitively, from a history $\tau$ such that $h(\tau) \in \Pre(X)$
the player can choose an action $a$ such that $(h(\tau),a) \in X$,
and if $(h(\tau),a) \in \Pre(X)$ then for all moves~$c$
chosen by the environment and supported by action $a$
we have  $h(\tau c) \in X$. 
In a game of perfect information, iterating the predecessor operator
from the target set $P_F$ until obtaining a fixpoint 
$X_{*} = \mu X. \Pre(X) \cup P_F$
gives the (existentially and universally) winning states,
that is $P_{\forall} = P_{\exists} = X_{*}$.

In a game of imperfect information with randomized strategies for the player,
given $X_{i+1} = \Pre(X_i)$,
from a history $\tau$ such that $h(\tau) \in X_{i+1}$,
we may consider playing all actions $a$ such that $(h(\tau),a) \in X_i$
uniformly at random. However, an action played in a history 
$\tau$ is also played in every indistinguishable history $\tau' \sim \tau$. Therefore,  
we need to ensure that for all actions $a$ played in $\tau$
(such that $(h(\tau),a) \in X_i$), playing $a$ in $\tau'$
does not leave $X_{*}$. Hence for $p =  h(\tau)$,
even if $(p,a) \in X_{*}$, the action $a$ should not
be played from $\tau$ if there exists $p' \approx p$
such that $(p',a) \not\in X_{*}$.

In our algorithm, we remove from $X_{*}$ all elements $(p,a)$
such that $(p',a) \not\in X_{*}$ for some $p' \approx p$,
that is we replace $X_{*}$ by its interior.
We define $$Y_{*} = \nu Y. \, \mu X. \,\, \interior_{\approx}(Y) \cap (\Pre(X) \cup P_F)  $$
and we show that the sets $Y_{*}$, $P_{\forall}$, and $P_{\exists}$
coincide. Note that the $\mu$-calculus formula for $Y_{*}$ is well defined
since $\Pre(\cdot)$ and $\interior_{\approx}(\cdot)$ are monotone operators, 
and that the fixpoint can be effectively computed since $P$ is finite.

As the fixpoint satisfies  $Y_{*} = \interior_{\approx}(Y_{*}) \cap (\Pre(Y_{*}) \cup P_F)$,
it follows that $Y_{*} = \interior_{\approx}(Y_{*})$ is closed under $\approx$,
and $Y_{*} \subseteq \Pre(Y_{*}) \cup P_F$. Moreover it is easy
to verify from the fixpoint iteration that $P_F \subseteq Y_{*}$ 
since $P_F$, which is closed under $\approx$, is always contained
in the least fixpoint, and never removed from the greatest fixpoint.

\begin{theorem}\label{theo:fixpoint}
The abstract states in the fixpoint $Y_*$ are the existentially winning states,
which coincide with the universally winning states,
$Y_* \cap P = P_{\forall} = P_{\exists}$.
\end{theorem}

Theorem~\ref{theo:fixpoint} entails the decidability of the existence
of an almost-sure winning strategy in games with reachability objective:
if there exists an almost-sure winning strategy (from $\epsilon$),
then $h(\epsilon) \in P_{\exists} = Y_*$, and conversely if $h(\epsilon) \in Y_*$
then $h(\epsilon) \in P_{\forall}$ since $h(\epsilon) \in P$,
and there exists an almost-sure winning strategy from $[\epsilon]_{\sim} = \{\epsilon\}$.

\begin{corollary}
There exists an almost-sure winning strategy if and only if $h(\epsilon) \in Y_*$.
\end{corollary}

The result of Theorem~\ref{theo:fixpoint} follows from the chain
of inclusions  $Y_* \cap P \subseteq P_{\forall}$ (Lemma~\ref{lem:Pforall}),
the already established $P_{\forall} \subseteq P_{\exists}$, and
$P_{\exists} \subseteq Y_* \cap P$ (Lemma~\ref{lem:Pexists}),

To show that $Y_* \cap P \subseteq P_{\forall}$ we construct 
a strategy for the player that is almost-sure winning from all $\tau$ (and from $[\tau]_{\sim}$)
such that $h(\tau) \in  Y_*$. The strategy plays uniformly at random 
all actions that ensure the successor $\tau c$ of $\tau$ remains
in $Y_*$, more precisely $h(\tau c) \in Y_*$. 
We show that at least one such action ensures progress towards reaching
a target state in $P_F$, thus with probability at least $\frac{1}{\abs{A}}$.
The target is reachable in at most $\abs{P}$ steps, which entails
bounded probability (at least $\nu = \frac{1}{\abs{A}^{\abs{P}}}$) to reach
$P_F$ (from every state of $Y_*$) and since the strategy ensures that $Y_*$ is never left, 
the probability of \emph{not} reaching $P_F$ is at most $(1-\nu)^k$ 
for all $k \geq 0$, and as $(1-\nu)^k \to 0$ when $k \to \infty$
the probability to eventually reach $P_F$ is $1$.


\begin{restatable}{lemma}{Pforall}
\label{lem:Pforall}%
$Y_*  \cap P \subseteq P_{\forall}$.
\end{restatable}

The last inclusion of the chain is proved by showing that there exists
a fixpoint that contains $P_{\exists}$ and since $Y_*$ is defined as 
the greatest fixpoint, we have $P_{\exists} \subseteq Y_*$.


\begin{restatable}{lemma}{Pexists}
\label{lem:Pexists}%
   $P_{\exists} \subseteq Y_* \cap P$.
\end{restatable}

The fixpoint computation for $Y_*$ can be implemented by a quadratic
algorithm, with respect to the number $\abs{P}$ of abstract states
(assuming a constant number of actions).

Both stochastic games on graphs and B\"uchi objectives
are subsumed by the results of this paper: in stochastic games, 
almost-sure B\"uchi reduces to almost-sure reachability~\cite[Lemma~8.3]{BGB12},
and games with stochastic transitions can be simulated by (non-stochastic) games 
with imperfect information~\cite[Theorem~5]{CDGH15}. These results can easily be lifted
to the more general framework of games with imperfect information defined
by indistinguishability relations.
On the other hand, we recall that almost-sure coB\"uchi is undecidable,
already in probabilistic automata with pure strategies~\cite{BGB12},
and also with randomized strategies~\cite[Corollary~2]{CDGH15},
but without the assumption that the coloring function is information consistent.
For the special case of information-consistent coB\"uchi objectives, the decidability status
of the synthesis problem for almost-sure winning is open (even in games {\it \`a la Reif}), 
to the best of our knowledge.
Finally for safety objectives, almost-sure winning 
is equivalent to sure winning (which requires that all possible outcomes 
satisfy the objective), and the problem is equivalent to synthesis with pure strategies~\cite{BDS23}.


\begin{theorem}\label{theo:as-reach}
Given a regular game $\G$ with imperfect information and a rectangular morphism $h: \Gamma^* \to P$ for $\G$,
the synthesis problem for almost-sure reachability and B\"uchi objectives can be solved in time $O(\abs{P}^2)$.
\end{theorem}

Since $\abs{P}$ is already exponential in games {\it \`a la Reif} (for which 
the synthesis problem is EXPTIME-complete),
the quadratic blow-up is not significant. 
For example, we get a $(k+1)$-EXPTIME complexity upper bound 
for the synthesis problem in the model of full-information 
protocols (FIP) with $k$ observers, using the rectangular morphism of~\cite[Section~5.2]{BDS23}.
The model of full-information protocols as presented in~\cite{BDS23}
features explicit communication actions that entail full disclosure
of all available information of the sender. The game involves
several players who may communicate, but only one active player who 
takes control actions, which reduces the synthesis problem to
a game with imperfect information as considered here.

A matching $(k+1)$-EXPTIME lower bound can be obtained by a straightforward adaptation
of the proof of~\cite[Theorem~3]{BDS23}, which presents a reduction from 
the membership problem for alternating $k$-EXPSPACE Turing machines
to the synthesis problem of (pure) winning strategies.
The reduction is such that in the constructed game, randomization does not help
the player. If the player takes a chance in deviating from the faithful simulation of
the Turing machine, a losing sink state is reached, thus with positive probability.

\begin{theorem}\label{theo:FIP-synthesis-complexity}
The synthesis problem for FIP games with $k$ observers (and almost-sure reachability objective) 
is $(k+1)$-EXPTIME-complete.
\end{theorem}

The construction of an almost-sure winning strategy $\straa$ can be done by following
the first step in the proof of Lemma~\ref{lem:Pforall}, which constructs $\straa$
given the value of the fixpoint $Y_*$. It follows that the constructed strategy 
is a regular function that can be represented by the automaton underlying 
the rectangular morphism of the game, and thus memory of size $\abs{P}$
is sufficient to define an almost-sure winning strategy.

\section{Reductions}

The results of Section~\ref{sec:as-reachability} may suggest the existence of 
a (strong) correspondence, possibly an equivalence, between the game $\G$
with regular indistinguishability relation and a simpler (finite-state) game $\H$ (essentially the 
abstract game presented in Section~\ref{sec:regular-games})
that would hold for arbitrary objectives.

While it is virtually impossible to establish the non-existence of a reasonably
strong such correspondence, we show for a natural notion of game 
equivalence that such a correspondence does not hold in general. 

\subsection{Alternating probabilistic trace equivalence}

Inspired by the notion of alternating refinement relations~\cite{AHKV98}
in non-stochastic game graphs with perfect information (also called alternating transition systems),
we consider the following definition of \emph{game refinement}.

Given two games $\G$ and $\H$, we say that $\G$ \emph{reduces} to $\H$ (denoted by $\G \preceq \H$)
if for all strategies $\straa_{\G}$ of the player in $\G$, 
there exists a strategy $\straa_{\H}$ of the player in $\H$, such that 
for all strategies $\strab_{\H}$ of the environment in $\H$, 
there exists a strategy $\strab_{\G}$ of the environment in $\G$, such that 
the probability measures $\Prb^{\straa_{\G},\strab_{\G}}$ in $\G$
and $\Prb^{\straa_{\H},\strab_{\H}}$ in $\H$ coincide (on all objectives $L \subseteq C^{\omega}$, that is 
$\Prb^{\straa_{\G},\strab_{\G}}(L) = \Prb^{\straa_{\H},\strab_{\H}}(L)$ for all $L$), in short:

$$\forall \straa_{\G} \cdot \exists \straa_{\H} \cdot \forall \strab_{\H} \cdot \exists \strab_{\G}
\cdot \forall L \subseteq C^{\omega}: \Prb^{\straa_{\G},\strab_{\G}}(L) = \Prb^{\straa_{\H},\strab_{\H}}(L).$$

As expected, the quantifications over strategies of the player range 
over information-consistent randomized strategies (in their respective game),
and the quantification over $L$ ranges over measurable sets.
Under theses quantifications, the condition 
$\Prb^{\straa_{\G},\strab_{\G}}(L) = \Prb^{\straa_{\H},\strab_{\H}}(L)$
can be replaced by $\Prb^{\straa_{\G},\strab_{\G}}(L) \leq \Prb^{\straa_{\H},\strab_{\H}}(L)$
(consider the complement of $L$, which is also a measurable set),
and then game refinement can be interpreted as the game $\H$ being 
easier\footnote{\label{note-words}The words ``easier'',  ``better'', ``lower'', etc. are interpreted in a non-strict sense.} than $\G$
for the player because for every strategy of the player in $\G$ there is a better\footnoteref{note-words}
strategy for the player in $\H$ (as reflected by the quantifier sequence $\forall \straa_{\G} \cdot \exists \straa_{\H}$),
which means that the environment can always ensure a lower\footnoteref{note-words} value in $\G$ (against $\straa_{\G}$)
than in $\H$ (against $\straa_{\H}$), as reflected by the quantifier sequence $\forall \strab_{\H} \cdot \exists \strab_{\G}$).

Note that in the special case of pure strategies (and non-stochastic games)
where the probability measures assign probability~$1$ to a single play, 
game refinement boils down to $$\forall \straa_{\G} \in \Straa^{\pure}_{\G}
\cdot \exists \straa_{\H} \in \Strab^{\pure}_{\H}: \outcome_{\H}(\straa_{\H}) \subseteq \outcome_{\G}(\straa_{\G})$$
where $\outcome_X(\straa)$ denotes the set of plays compatible with $\straa$ in game $X$.
The condition $\outcome_{\H}(\straa_{\H}) \subseteq \outcome_{\G}(\straa_{\G})$
is equivalent to $$\forall \strab_{\H} \in \Strab^{\pure}_{\H} \cdot \exists \strab_{\G} \in \Strab^{\pure}_{\G}: 
\outcome_{\H}(\straa_{\H},\strab_{\H}) = \outcome_{\G}(\straa_{\G},\strab_{\G}),$$
where $\outcome_X(\straa,\strab)$ denotes the set of plays compatible with $\straa$ and $\strab$ in game $X$.
This is essentially the definition of alternating trace containment~\cite[Section~4]{AHKV98}.
Hence our definition of game refinement is a natural generalization of both 
alternating trace containment to a stochastic and imperfect-information setting, 
and of refinement for labeled Markov decision processes~\cite[Section~4]{DHR08} to games.

We say that two games $\G$ and $\H$ are \emph{equivalent} (or inter-reducible)
if $\G \preceq \H$ and $\H \preceq \G$. We sometimes refer to this 
equivalence as \emph{alternating probabilistic trace equivalence}.

\subsection{Bijection and bisimulation}\label{sec:bij-bis}

We compare the approach in Section~\ref{sec:as-reachability} and the technique
for solving almost-sure reachability in games with partial observation
{\it \`a la Reif}~\cite{CDHR07}. Recall that games with partial observation
are a special case of games with regular indistinguishability relation.

The solution of the synthesis problem with randomized strategies 
for games {\it \`a la Reif} established a bijection (denoted by~$h$, but we call it $\hat{h}$ 
to emphasize that it is a mapping between sequences) between the histories in the original
game $\G$ and in a game $\H$ of partial observation played as follows: given 
a history~$\rho$ (initially~$\epsilon$), the player chooses an action $a \in A$,
and the environment extends the history to $\rho' = \hat{h}(\tau c)$ 
where $\hat{h}(\tau) = \rho$ and $c$ is a move (in $\G$) such that $\act(c) = a$.
Two histories are indistinguishable if all their respective prefixes (of the same length)
are $\approx$-equivalent (according to the definition of $\approx$ before Lemma~\ref{lem:approx-equiv}). 
We call $\H$ the \emph{abstract game induced} from $\G$ by $\hat{h}$, and by an 
abuse of notation, we denote by $\approx$ the indistinguishability relation of $\H$.

The bijection $\hat{h}$ naturally extends to a bijection between plays in $\G$ and plays in $\H$,
and further to a bijection between strategies in $\G$ and strategies in $\H$
that preserves probability measures. 
The existence of this bijection immediately ensures that the games $\G$ and $\H$ are equivalent
(in the sense of alternating probabilistic trace equivalence).
Another important consequence is that the information trees of $\G$ and $\H$ are
isomorphic, a situation that is not guaranteed when $\G$ is a game with regular 
indistinguishability relation and $\H$ is a finite-state game (i.e., the set 
$h(\Gamma^*) = \{\last(\hat{h}(\tau)) \mid \tau \in \Gamma^* \}$ is finite). 
Indeed, the branching degree of the information tree of $\G$ may be unbounded~\cite{BD20},
whereas the branching degree of the information tree of $\H$ is necessarily bounded 
since $\H$ is a finite game. It is therefore impossible to follow
the same route as in games {\it \`a la Reif}, using bijections.

Since bijection is too strong for our general setting, a more realistic hope is to rely on 
the existence of a bisimulation between the information tree of $\G$ and the information tree of $\H$~\cite[Theorem~1]{BDS23}.
Note that bisimulation is insensitive to the branching degree. 
The bisimulation ensures that, under pure strategies for the player, the (perfect-information) 
games played on the information trees are equivalent (in the sense
of alternating trace equivalent). The equivalence induced by the bisimulation 
is such that the transfer of strategies (the construction of $\straa_{\H}$ given $\straa_{\G}$)
is of the form $\straa_{\H}(\cdot) = \straa_{\G}(\mu(\cdot))$ where $\mu$ is a mapping between 
bisimilar histories\footnote{In games played over information trees, the histories are sequences of information sets; 
we do not denote them by a symbol to avoid having to define them formally.}
that is sequential, that is the image of a prefix of a history
is a prefix of the image of that history.
So, for pure strategies the action played by $\straa_{\H}$ at a history is a copy of 
the action played by $\straa_{\G}$ in a bisimilar history in $\G$. 
In the case of  randomized strategies, the player chooses a distribution over actions,
which we may view as attaching probabilities to actions, in a way that could be 
copied along with the action to transfer randomized strategies (on one hand from $\G$ to $\H$,
and on the other hand from $\H$ to $\G$) and establish alternating probabilistic trace 
equivalence (containment in both directions) of the two games.

We recall the strong relationships between games with imperfect information
and their induced abstract game (\figurename~\ref{fig:reductions}).
First, there is a natural bijection between an imperfect-information game $\G$ 
and the perfect-information game played on its information tree $\U_{\sim}(G)$,
which is a witness of alternating trace equivalence of the two games~\cite[Lemma~1]{BDS23}.
Second, assuming the existence of a rectangular morphism $h$ for~$\G$, the information
trees of $\G$ and of $\H$ (the abstract game induced from $\G$ by $h$) are bisimilar.
The bijections and the bisimulation (\figurename~\ref{fig:reductions}) entail
alternating trace equivalence of the games $\G$ and $\H$ and their information trees,
and thus equivalence of the synthesis problems with pure strategies in $\G$ and in~\mbox{$\H$~\cite[Lemma~2,Theorem~1]{BDS23}}.

\begin{figure}[!t]
\hrule
\medskip
\begin{center}
    \input{figures/reductions.tex}

\end{center} 
\hrule
 \caption{Equivalences between the game $\G$, its induced abstract game $\H$, and their information trees.  \label{fig:reductions}}
\end{figure}
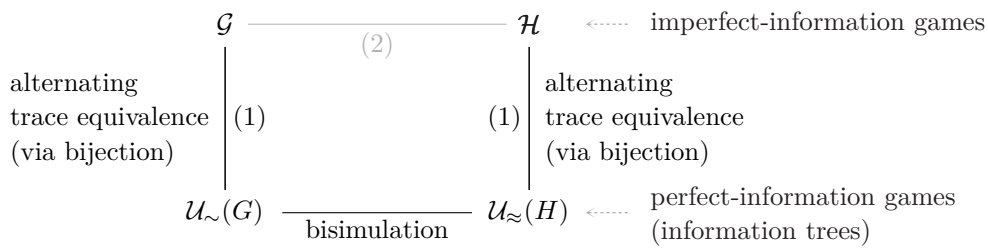

\subsection{Counterexamples}\label{sec:cex}

We present two examples showing the difficulty to adapt the results
of alternating trace equivalence for games with pure strategies,
to alternating probabilistic trace equivalence for games with randomized strategies.

First, we show that with randomized strategies, playing on the information tree 
of a game with imperfect information is not equivalent to playing in the original game,
although this is true for pure strategies (links (1) of \figurename~\ref{fig:reductions}).

The first example is matching pennies (the game $\G$ in \figurename~\ref{fig:game-as}). 
Recall that the set of actions is $A = \{a,b\}$, the set of moves is $\Gamma = \{a_1,a_2,b_1,b_2\}$
with $\act(a_i) = a$ and $\act(b_i) = b$ for $i=1,2$.
It is well known that the player cannot win matching pennies using a pure strategy,
but wins almost-surely by choosing head and tail uniformly at random in every round.
However, in the game played on the information tree of $\G$ (\figurename~\ref{fig:game-as-information-tree}),
the player cannot win, even with positive probability. 
After the first round the information set is $U_2 = \Gamma$, which has two successors
on both actions $a$ and $b$. This is because $U_2$ contains both a history $\tau$
such that $\tau a_i$ is winning and a history $\tau'$ such that $\tau' b_i$ is winning,
hence $U_3$ is a successor of $U_2$ (on both actions $a$ and $b$), and
analogously $X$ is a successor of $U_2$ ($X$ is a placeholder for the whole tree of 
\figurename~\ref{fig:game-as-information-tree}).
Therefore, after any choice of action by the player in $U_2$, the environment
can always choose $X$ and continue the game, ensuring that the reachability 
objective is never satisfied. 

This first example illustrates a crucial difference between pure and randomized
strategies. Considering a game $\G$ (with imperfect information induced
by an indistinguishability relation~$\sim$) and
the game played on its information tree $\U_{\sim}(G)$ (which is of perfect information), 
there is a natural bijection between (information-consistent) strategies in $\G$
and strategies in $\U_{\sim}(G)$. The bijection exists for both pure and randomized
strategies. However, only in the case of pure strategies, the bijection is
a witness of alternating trace equivalence (\figurename~\ref{fig:reductions}).
Intuitively, this is because once a \emph{pure} strategy for the player
is fixed, the environment ``knows'' the specific action played in every history (which is the same
within an information set),  and therefore constructing a path in the original
game $\G$ compatible with the player strategy is equivalent to constructing a path
in $\U_{\sim}(G)$ (informally, the information tree is a valid abstraction for ``existence
of paths''). In contrast, once a \emph{randomized} strategy for the player
is fixed, the environment only  ``knows'' the specific distribution over actions
played in every history, which leaves an element of surprise as to which action
will be played when the strategy is executed. In this context, the information 
tree is no longer a valid abstraction, intuitively because it allows the environment 
to choose a history within the current information set \emph{after} the randomly chosen
action is known, as illustrated by matching pennies (\figurename~\ref{fig:game-as-information-tree}).
To be valid, the abstraction should force the environment to stick to a single
choice of a history, before the action is drawn from the distribution chosen by the player.
We may consider the abstract game $\H$ induced from $\G$ as a candidate for a 
richer abstraction, and establish a direct link (2) between $\G$ and $\H$ (\figurename~\ref{fig:reductions}),
given the link~(1) breaks. Unfortunately,
in our second example $\H$ is not equivalent to $\G$ for alternating probabilistic trace 
equivalence ($\G \not\preceq \H$), showing the absence of such a direct link,
which however exists in partial-observation games {\it \`a la Reif}, 
via the bijection $\hat{h}$ mentioned in Section~\ref{sec:bij-bis}.


\begin{figure}[!t]%
\begin{center}
\hrule
\medskip
\subfloat[The game $\G$.]{
   \input{figures/game-cex.tex}
   \label{fig:game-cex}
}
\hfill
\subfloat[The game $\H$.]{
   \input{figures/game-cex-mapped.tex}
    \label{fig:game-cex-mapped}
}
\hfill
\mbox{}\\
\bigskip
\subfloat[The game $\G$ with $x_1 = \frac{1}{2}$, $x_2 = 1$, and $x_3 = 0$.]{
   \input{figures/game-cex-instantiated.tex}
   \label{fig:game-cex-instantiated}
}
\hfill
\subfloat[The game $\H$ with $y_1 = \frac{1}{2}$, $z_1 = 1$, and $z_2 = 1-y_2$.]{
   \input{figures/game-cex-mapped-instantiated.tex}
    \label{fig:game-cex-mapped-instantiated}
}
\hrule
\caption{A game $\G$ and the abstract game $\H$ induced by a rectangular morphism,
with randomized strategies encoded by the variables $\bar{x},\bar{t}$ (in $\G$) and $\bar{y},\bar{z}$ (in $\H$).
 \label{fig:figgame-cex}}%
\end{center}
\end{figure}
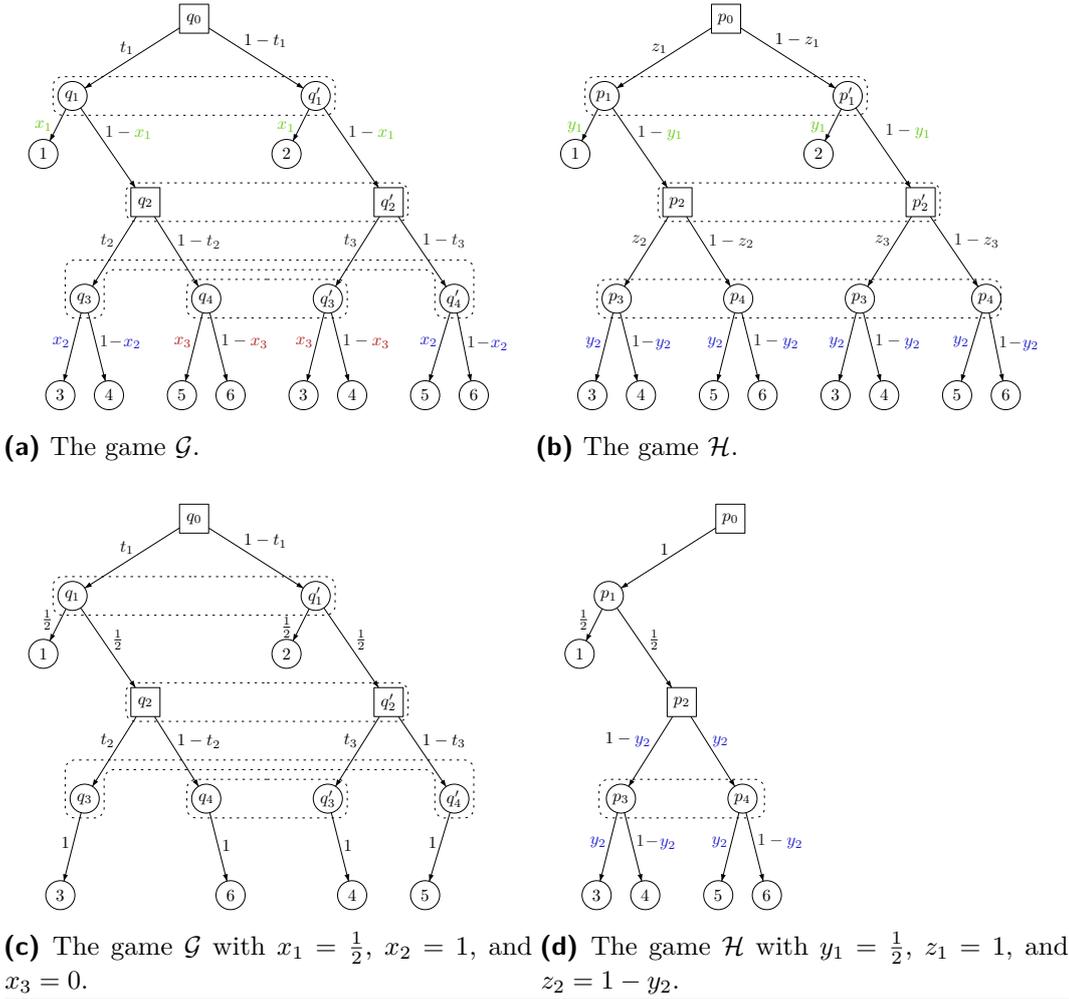

The second example (\figurename~\ref{fig:game-cex}) shows a game $\G$ 
and its abstract game $\H$ induced by a rectangular morphism. 
To avoid tedious description of the game,
we present the key features informally. The game $\G$ has two actions 
for the player and two possible responses for the adversary.
The actions and moves are not shown in \figurename~\ref{fig:game-cex},
and states are depicted as circles when only the action of the player 
determines the successor, and as boxes when only the move of the environment
determines the successor. All states have color $0$ except the leaves, 
labeled by their color $1,2,\dots,6$. The leaves are sink states (with a self-loop
on every move). The dashed lines show the indistinguishability classes.
It is easy to check that the coloring function is information-consistent.

\figurename~\ref{fig:game-cex-mapped} shows the abstract game $\H$,
as a tree with the same shape as in \figurename~\ref{fig:game-cex} for $\G$.
The value $h(\tau)$ of a history can be read as the label of the node
in \figurename~\ref{fig:game-cex-mapped} corresponding to the path defined
by $\tau$ in \figurename~\ref{fig:game-cex}. For example, 
the histories (of length $3$) in $\G$ that correspond to the nodes $q_3$ and $q'_3$
are mapped by $h$ to the abstract state $p_3$. \figurename~\ref{fig:game-cex-mapped}
can be viewed as the unraveling of the morphism $h$. It is easy to verify
that $h$ is a rectangular refinement of the coloring function, 
and that $\approx$-equivalent abstract states are grouped within the
dashed lines.

We show that $\G$ does not reduce to $\H$ according to alternating probabilistic trace 
containment ($\G \not\preceq \H$). We describe all randomized strategies 
of the player and of the environment by specifying the probability to play 
an action or a move in each (relevant, i.e., non-leaf) history.
We use variables $\bar{x} = x_1,x_2,x_3$ for $\straa_{\G}$; $\bar{y} = y_1,y_2$ for $\straa_{\H}$;
$\bar{z} = z_1,z_2,z_3$ for $\strab_{\H}$; and $\bar{t} = t_1,t_2,t_3$ for $\strab_{\G}$.
Note that the player may only use information-consistent strategies,
hence the same variable is used from indistinguishable histories to
describe $\straa_{\G}$ and $\straa_{\H}$.
The claim $\G \not\preceq \H$ can be written as

$$ \psi \equiv \exists \bar{x} \cdot \forall \bar{y} \cdot \exists \bar{z} \cdot \forall \bar{t}: 
\bigvee_{c = 1}^{6} \varphi_c$$

where $\varphi_c$ is expresses that the mass of probability on color $c$ differs
in $\G$ and in $\H$. For example, 
$$\varphi_1 \equiv \,  t_1 x_1 \neq z_1 y_1, \text{and}$$
$$\varphi_3 \equiv \, t_1 t_2 x_2 (1-x_1) + t_3 x_3 (1-t_1) (1-x_1) \neq 
z_1 z_2 y_2 (1-y_1) + z_3 y_2 (1-z_1) (1-y_1).$$
 
To show that $\psi$ holds, we fix $x_1 = \frac{1}{2}$, $x_2 = 1$, and $x_3 = 0$.
Now, consider all possible values of $\bar{y}$ and observe that if $y_1 \neq \frac{1}{2}$,
then $\psi$ holds since the probability mass in colors $\{1,2\}$ is $\frac{1}{2}$
in $\G$ and $y_1$ in $\H$ (and thus $\varphi_1 \lor \varphi_2$ must hold regardless
of the value of $\bar{z}$ and $\bar{t}$). So, it remains to show that $\psi$ holds
for $y_1 = \frac{1}{2}$ (and for all values of $y_2$). Fix $z_1 = 1$ and $z_2 = 1-y_2$
(the value of $z_3$ is arbitrary), as illustrated in \figurename~\ref{fig:game-cex-instantiated} 
and \figurename~\ref{fig:game-cex-mapped-instantiated}. 

Consider all possible values of $\bar{t}$ and observe that if $t_1 \neq 1$,
then $\psi$ holds since the probability mass in color~$1$ is~$\frac{1}{2}$
in $\G$ and $t_1$ in $\H$ (thus $\varphi_1$ holds). With $t_1 = 1$,
we show that $\varphi_4 \lor \varphi_5$ holds, that is:
$$ 0 \neq \frac{(1-y_2)^2}{2} \ \lor \  0 \neq \frac{(y_2)^2}{2} $$
which holds since either $y_2 \neq 0$ or $1-y_2 \neq 0$.

We conclude that $\psi$ holds and thus $\G$ does not reduce to $\H$ ($\G \not\preceq \H$).
Note that slightly stronger statements can be proved by a similar argument:
first, the strategy $\strab_{\H}$ (represented by the variables $\bar{z}$) can be 
chosen pure, by setting $z_2 = 1$ if $y_2 = 0$ and $z_2 = 0$ if $y_2 = 1$ (and setting 
$z_2$ to an arbitrary value in $\{0,1\}$ otherwise);
second, the example does not even preserve almost-sure probabilities, in the following sense.
Given a target set $T \subseteq C = \{1,\dots,6\}$ of colors,
let $\Reach(T) = \bigcup_{k \in T} C^*kC^{\omega}$ be the reachability objective
with target $T$. Consider the formula

$$\psi' \equiv \exists \bar{x} \cdot \forall \bar{y} \cdot \exists \bar{z} \cdot \forall \bar{t}: 
\bigvee_{T \subseteq C} \lnot (\Prb_{\G}^{\bar{x},\bar{t}}(\Reach(T)) = 1 \Leftrightarrow
\Prb_{\H}^{\bar{y},\bar{z}}(\Reach(T)) = 1)$$
where $\Prb_{\G}^{\bar{x},\bar{t}}(\Reach(T))$ is the probability that the
objective $\Reach(T)$ is satisfied in $\G$ under strategies $\straa_{\G}$ defined by the variables 
$\bar{x}$ for the player and $\strab_{\G}$ defined by $\bar{t}$ for the environment
(and analogously for $\Prb_{\H}^{\bar{y},\bar{z}}(\Reach(T))$ in $\H$).
It is easy to check that the condition $\psi'$ entails~$\psi$.

To prove that $\psi'$ holds, we fix $x_1 = \frac{1}{2}$, $x_2 = 1$, and $x_3 = 0$
as before (\figurename~\ref{fig:game-cex-instantiated}). 
Consider three possible cases for $\bar{y}$: $(i)$ if $y_1 =0$,
then no matter the value of the other variables the objective 
$\Reach(\{1,2\})$ has probability $\frac{1}{2}$ in $\G$ and probability $0$ in $\H$,
so we take $T = C \setminus \{1,2\} = \{3,4,5,6\}$;
$(ii)$ if $y_1 = 1$, we take $T = \{1,2\}$ (the probability mass is $\frac{1}{2}$ in $\G$ and $1$ in $\H$);
$(iii)$ otherwise $0 < y_1 < 1$, and we take $z_1 = 1$ and $z_2 = 1-y_2$ (illustrated in \figurename~\ref{fig:game-cex-mapped-instantiated} for $y_1 = \frac{1}{2}$), and consider all possible values of $\bar{t}$:
if $t_1 \neq 1$ then $\Reach(\{2\})$ has positive probability in $\G$ and probability $0$ in $\H$,
so we take $T = C \setminus \{2\} = \{1,3,4,5,6\}$, and if $t_1 = 1$, then 
$\Reach(\{4,5\})$ has probability $0$ in $\G$ and positive probability in $\H$,
so we take $T = C \setminus \{4,5\} = \{1,2,3,6\}$.

\begin{theorem}\label{theo:cex}
There exists a regular game $\G$ such that the abstract game $\H$ induced from $\G$
is not alternating probabilistic trace equivalent to $\G$. 
\end{theorem}

\input{main.bbl}
\bibliographystyle{plainurl}

\newpage
\appendix

\section{Proof of Lemma~\ref{lem:Pforall}}

The abstract states computed by the fixpoint $Y_*$ are all universally winning.

\Pforall*

\begin{proof}  
Given $p \in Y_* \cap P$, we show that $p \in P_{\forall}$. 
Let $A_p = \{a \in A \mid (p,a) \in Y_*\}$ and show that $A_p \neq \emptyset$.
Since $Y_{*} \subseteq \Pre(Y_{*}) \cup P_F$,
either $(i)$ $p \in P_F$ and then $A_p = A \neq \emptyset$ as $P_F \times A \subseteq \Pre(P_F)$
(states in $P_F$ form a sink set and $P_F$ is closed under $\approx$); 
or $(ii)$ $p \in \Pre(Y_{*})$ and by definition 
there exists $a \in A$ such that $(p,a) \in Y_*$, hence $A_p \neq \emptyset$.

Consider the strategy $\straa$ for the player defined, for all $\tau \in \Gamma^*$,
by $\straa(\tau) = d_U(A_p)$
where $p = h(\tau)$ and $d_U(\cdot)$ is the uniform distribution. We show that
$\straa$ is information-consistent. Given $\tau' \sim \tau$, let $p = h(\tau)$
and $p' = h(\tau')$. Hence we have $p \approx p'$, and since $Y_*$ is closed
under $\approx$, if $(p,a) \in Y_*$ then  $(p',a) \in Y_*$. It follows that 
$A_p = A_{p'}$ and thus  $\straa(\tau) = \straa(\tau')$.

The strategy $\straa$ ensures that $Y_*$ is never left: if $h(\tau) \in Y_*$
and $a \in \Supp(\straa(\tau))$, then  $h(\tau c) \in Y_*$
for all $c$ such that $\act(c) = a$. Indeed, for $p = h(\tau)$
we have $a \in A_p$ and $(p,a) \in Y_* $ thus $(p,a) \in \Pre(Y_{*})$
which imply $\delta^P(p,c) = h(\tau c) \in Y_* $.

We now show that $\straa$ ensures from $Y_*$ a target state is reached
within the next $N = \abs{P}$ steps with probability at least $\nu = \frac{1}{\abs{A}^N}$,
that is $\Prb_{\tau}^{\straa,\strab}(\Reach^{\leq \abs{\tau} + N}) \geq \nu$ 
for all strategies $\strab$ of the environment and
for all $\tau \in \Gamma^*$ such that $h(\tau) \in Y_*$.

Since $Y_* = \mu X. \, \interior_{\approx}(Y_*) \cap (\Pre(X) \cup P_F) $,
the set $Y_*$ is the limit of the (non-decreasing) sequence $(X_i)_{i\in \nat}$ where
$X_0 = \emptyset$ and $X_{i+1} = Y_* \cap (\Pre(X_i) \cup P_F)$.
The \emph{rank} of an element $y \in Y_*$ is the least $i \geq 0$ such that $y \in X_i$ 
(i.e., such that $y \in X_i$ and $y \not\in X_{i-1}$). 
By extension, the rank of a history $\tau$ such that $h(\tau) \in Y_*$
is the rank of $h(\tau)$.
Note that states $p \in Y_* \cap P$ have odd rank, and the state-action
pairs  $(p,a) \in Y_* \cap (P \times A)$ have even rank. 
The smallest rank is $1$ and corresponds to the set $P_F$ of target states.  
Let $N_*$ be the largest rank, which is bounded by $\abs{Y_*}$. 

Intuitively, from any history $\tau$ such that $h(\tau) \in Y_*$ the strategy 
$\straa$ ensures, against all strategies of the environment, 
that the history $\tau c$ after the next round has a lower rank
(unless the rank of $\tau$ is $1$, and thus $\tau$ is a winning history)
with probability at least $\frac{1}{\abs{A}}$. 
Let $r > 1$ be the rank of $h(\tau)$, then  $h(\tau) \in \Pre(X_{r-1})$ and thus 
there is an action $a \in A$ such that $(h(\tau), a) \in X_{r-1}$, thus 
$(h(\tau), a) \in \Pre(X_{r-2})$ and $h(\tau c) \in X_{r-2}$
for all $c$ with $\act(c) = a$. It follows that $a \in A_p$ is played
with probability at least $\frac{1}{\abs{A}}$ and the intuitive claim follows.

Combined with the fact that the strategy $\straa$ ensures $Y_*$ is never left, 
this claim shows that from all $\tau$ such that $h(\tau) \in Y_*$, against all 
strategies $\strab$ of the environment, the reachability objective is satisfied 
(i.e., the rank decreases to $1$) within the next $N_*$ rounds with probability 
at least $\nu = \frac{1}{\abs{A}^{N_*}}$.
Formally $\Prb_{\tau}^{\straa,\strab}(\Reach^{\leq \abs{\tau} + N_*}) \geq \nu$
against all strategies. 
It follows that the probability of \emph{not} reaching a target state
within $k \cdot N_*$ rounds is at most $(1-\nu)^k$ which tends to $0$
as $k \to \infty$, thus $\Prb_{\tau}^{\straa,\strab}(\Reach) = 1$ 
by the squeeze theorem. We conclude that the player is almost-sure winning 
from $[\tau]_{\sim}$ since $Y_{*}$ is closed under $\approx$, which 
implies that $p = h(\tau) \in P_{\forall}$.
\end{proof}

\newpage

\section{Proof of Lemma~\ref{lem:Pexists}}

The fixpoint $Y_*$ contains all existentially winning states.

\Pexists*

\begin{proof}
As a preliminary, it will be useful to note that if a strategy $\straa_{{\sf as}}$ 
for the player is almost-sure winning from an information set $[\tau]_{\sim}$,
then for all actions $a \in \Supp(\straa_{{\sf as}}(\tau))$,
for all moves $c$ supported by $a$, $\act(c) = a$,
the strategy $\straa_{{\sf as}}$ is almost-sure winning from $[\tau c]_{\sim}$. 
An equivalent conclusion is that 
for all actions $a \in \Supp(\straa_{{\sf as}}(\tau))$,
the pair $(h(\tau),a)$ is in $\interior_{\approx}(\Pre(P_{\exists}))~(\star)$.

We now proceed with the proof of the lemma.
First, $P_{\exists} \subseteq P$ by definition. To show that 
$P_{\exists} \subseteq Y_*$, since $Y_*$ is defined as a greatest fixpoint,
it is sufficient to show that $Y_{\exists} := P_{\exists} \cup \interior_{\approx}(\Pre(P_{\exists}))$ 
is a fixpoint, that is
$Y_{\exists} = \mu X. \, \interior_{\approx}(Y_{\exists}) \cap (\Pre(X) \cup P_F)$.
Since $P_{\exists}$ is closed under $\approx$ (i.e., $P_{\exists} = \interior_{\approx}(P_{\exists})$) so is $Y_{\exists}$,
and we only need to show that $Y_{\exists} \subseteq \mu X. \, Y_{\exists} \cap(\Pre(X) \cup P_F)  =: X_*$
(the converse inclusion is trivial).

Given $y \in Y_{\exists}$, we show that $y \in X_*$, which
concludes the proof. We consider two cases,
either $y = p \in P_{\exists}$ or $y = (p,a) \in \interior_{\approx}(\Pre(P_{\exists}))$:

\begin{itemize}

\item[$(i)$] if $p \in P_{\exists}$, then there exists $\tau_p \in \Gamma^*$ such that $h(\tau_p) = p$ and 
the player is almost-sure winning from $[\tau_p]_{\sim}$ using a strategy $\straa_{{\sf as}}$. 
Assume towards contradiction that $p \not\in X_*$. 

  
We construct a (pure) spoiling strategy $\strab$ for the environment as follows.
For all $\tau \in \Gamma^*$ such that $h(\tau) \in P_{\exists}$ and $h(\tau) \not\in X_*$ (thus $h(\tau) \not\in \Pre(X_{*})$),
for all actions $a \in \Supp(\straa_{{\sf as}}(\tau))$ played by~$\straa_{{\sf as}}$ in $\tau$, 
the pair $(h(\tau),a)$ is not in $X_*$, 
and further not in $\Pre(X_*)$ since $(h(\tau),a) \in \interior_{\approx}(\Pre(P_{\exists}))$ 
by~$(\star)$, thus there exists a move $c$ supported
by action $a$ such that $h(\tau c) = \delta^P(h(\tau),c) \not\in X_*$.
Define $\strab(\tau, a) = c$, and define $\strab(\tau', a)$ arbitrarily
for $\tau'$ such that $h(\tau') \in X_*$.

We have the following, for all histories $\tau \in \Gamma^*$ and 
for $a \in \Supp(\straa_{{\sf as}}(\tau))$ and $c = \strab(\tau, a)$:

\begin{itemize}
\item if $h(\tau) \in P_{\exists}$, then $h(\tau c) \in P_{\exists}$ (since $\straa_{{\sf as}}$ is almost-sure winning),
\item if $h(\tau) \in P_{\exists}$ and moreover $h(\tau) \not\in X_*$, then $h(\tau c) \not\in X_*$,
\item $h(\tau_p) \in P_{\exists}$ and $h(\tau_p) \not\in X_*$. 
\end{itemize}

It follows by an inductive argument that against $\straa$, the constructed 
strategy $\strab$ ensures from $\tau_p$ that $X_*$ and thus $P_F$ is never reached,
hence $\Prb_{\tau_p}^{\straa_{{\sf as}},\strab}(\Reach) = 0$, which contradicts that  
$\straa_{{\sf as}}$ is almost-sure winning from $\tau_p$.
Hence $p \in X_*$.

\item[$(ii)$] if $(p,a) \in \interior_{\approx}(\Pre(P_{\exists}))$, then
$\delta^P(p,c)) \in P_{\exists}$ for all moves $c$ supported by~$a$.
By an argument similar to case $(i)$, we have $\delta^P(p,c)) \in X_*$.
It follows that $(p,a) \in \Pre(X_*)$ and thus $(p,a) \in X_* =  Y_{\exists} \cap (\Pre(X_*) \cup P_F)$. 
\end{itemize}

\end{proof}

\end{document}

%% file: figures/game-as.tex
\begin{gpicture}(53,50)(0,0)


\gasset{Nw=6,Nh=6,Nmr=3}


\node[Nmarks=n, Nmr=0](q0)(25,45){$q_0$}

\node[Nmarks=n](q1)(10,25){$q_1$}
\node[Nmarks=n](q2)(40,25){$q_2$}
\node[Nmarks=r](q3)(10,5){$q_3$}
\node[Nframe=n,Nmarks=n](q4)(40,5){$x$}
\node[Nframe=n,Nmarks=n](label)(50,5){}
\node[Nframe=n,Nmarks=n](labelreal)(50,4.7){\gray50{$q_0$}}
\drawedge[ELside=r,ELpos=50, linegray=.6, AHangle=30, AHlength=0.2, dash={0.4 0.4}0](q4,label){}



\node[Nmarks=n, Nh=10, Nw=40,dash={0.5 1}0](box)(25,25){}

\drawedge[ELside=r,ELpos=50, ELdist=.5](q0,q1){$a_1,b_1$}
\drawedge[ELside=l,ELpos=50, ELdist=.5](q0,q2){$a_2,b_2$}

\drawedge[ELside=r,ELpos=50, ELdist=.5](q1,q3){$a_1,a_2$}
\drawedge[ELside=r,ELpos=30, ELdist=0](q1,q4){$b_1,b_2$}
\drawedge[ELside=l,ELpos=30, ELdist=0.5](q2,q3){$b_1,b_2$}
\drawedge[ELside=l,ELpos=50, ELdist=.5](q2,q4){$a_1,a_2$}



\end{gpicture}

%% file: figures/game-as-information-tree.tex
\begin{gpicture}(40,50)(0,0)


\gasset{Nw=8,Nh=8,Nmr=2}

\node[Nmarks=n](q1)(15,45){$U_1$}

\node[Nmarks=n](q2)(15,25){$U_2$}

\node[Nmarks=r](q3)(5,5){$U_3$}
\node[Nframe=n,Nmarks=n](q4)(25,5){$X$}

\node[Nframe=n,Nmarks=n](label)(37,5){}
\node[Nframe=n,Nmarks=n](labelreal)(37,4.7){\gray50{$U_1$}}
\drawedge[ELside=r,ELpos=50, linegray=.6, AHangle=30, AHlength=0.2, dash={0.4 0.4}0](q4,label){}

\drawedge[ELside=l,ELpos=50, ELdist=1](q1,q2){$a,b$}
\drawedge[ELside=r,ELpos=50, ELdist=.5](q2,q3){$a,b$}
\drawedge[ELside=l,ELpos=50, ELdist=.5](q2,q4){$a,b$}



\end{gpicture}

%% file: figures/reductions.tex
\begin{gpicture}(130,35)(0,0)


\gasset{Nw=6,Nh=6,Nmr=3}

\node[Nframe=n](nG)(30,30){$\G$}
\node[Nframe=n, Nw=15](nUG)(30,5){$\U_{\sim}(G)$}

\node[Nframe=n](nH)(70,30){$\H$}
\node[Nframe=n, Nw=15](nUH)(70,5){$\U_{\approx}(H)$}




\drawedge[ELside=l,ELpos=50, ELdist=1, AHnb=0](nG,nUG){$(1)$}
\drawedge[ELside=r,ELpos=50, ELdist=0, AHnb=0](nG,nUG){\begin{tabular}{l}alternating\\trace equivalence\\(via bijection)\end{tabular}}

\drawedge[ELside=r,ELpos=50, ELdist=1, AHnb=0](nUG,nUH){bisimulation}
\drawedge[ELside=l,ELpos=50, ELdist=0, AHnb=0](nH,nUH){\begin{tabular}{l}alternating\\trace equivalence\\(via bijection)\end{tabular}}
\drawedge[ELside=r,ELpos=50, ELdist=1, AHnb=0](nH,nUH){$(1)$}

\drawedge[ELside=r,ELpos=50, ELdist=1, AHnb=0, linegray=.7](nG,nH){\graydarker{$(2)$}}

\node[Nframe=n](label)(108,16.3){\begin{tabular}{l}imperfect-information games\\[+53pt] perfect-information games\\ (information trees) \end{tabular}}


\node[Nframe=n, Nw=15](nHbis)(70,30){$\H$}
\node[Nframe=n](label)(86,30){}
\drawedge[ELside=r,ELpos=50, linegray=.6, AHangle=30, AHlength=0.2, dash={0.4 0.4}0](label,nHbis){}
\node[Nframe=n](label)(86,5){}
\drawedge[ELside=r,ELpos=50, linegray=.6, AHangle=30, AHlength=0.2, dash={0.4 0.4}0](label,nUH){}



\end{gpicture}

%% file: figures/game-cex.tex
\scalebox{.64}{
\begin{gpicture}(100,86)(0,0)

\gasset{Nw=6,Nh=6,Nmr=3}


\node[Nmarks=n, Nmr=0](q0)(35,82){$q_0$}

\node[Nmarks=n](q1)(10,66){$q_1$}
\node[Nmarks=n](q1p)(60,66){$q'_1$}

\node[Nmarks=n](l1)(4,54){$1$}
\node[Nmarks=n](l2)(54,54){$2$}

\node[Nmarks=n](l3)(7.5,4){$3$}
\node[Nmarks=n](l4)(17.5,4){$4$}
\node[Nmarks=n](l5)(32.5,4){$5$}
\node[Nmarks=n](l6)(42.5,4){$6$}
\node[Nmarks=n](l3p)(57.5,4){$3$}
\node[Nmarks=n](l4p)(67.5,4){$4$}
\node[Nmarks=n](l5p)(82.5,4){$5$}
\node[Nmarks=n](l6p)(92.5,4){$6$}

\node[Nmarks=n, Nmr=0](q2)(25,44){$q_2$}
\node[Nmarks=n, Nmr=0](q2p)(75,44){$q'_2$}

\node[Nmarks=n](q3)(12.5,24){$q_3$}
\node[Nmarks=n](q3p)(62.5,24){$q'_3$}
\node[Nmarks=n](q4)(37.5,24){$q_4$}
\node[Nmarks=n](q4p)(88.5,24){$q'_4$}

\node[Nmarks=n, Nh=8, Nw=58, Nmr=2, dash={0.5 1}0](box)(35,66){}

\node[Nmarks=n, Nh=8, Nw=58, Nmr=2,dash={0.5 1}0](box)(50,44){}

\node[Nmarks=n, Nh=8, Nw=33, Nmr=2,dash={0.5 1}0](box)(50,24){}
\drawline[arcradius=2,AHnb=0,dash={0.5 1}0](18.5,32)(8.5,32)(8.5,20)(16.5,20)(16.5,30)(84.5,30)(84.5,20)(92.5,20)(92.5,32)(18.5,32)


\drawedge[ELside=r,ELpos=50, ELdist=.5](q0,q1){$t_1$}
\drawedge[ELside=l,ELpos=50, ELdist=.5](q0,q1p){$1-t_1$}


\drawedge[ELside=r,ELpos=61, ELdist=.5](q1,l1){\mygreen{{\small $x_1$}}}
\drawedge[ELside=r,ELpos=63, ELdist=.5](q1p,l2){\mygreen{{\small $x_1$}}}

\drawedge[ELside=l,ELpos=48, ELdist=.5](q1,q2){{\small $1-\mygreen{x_1}$}}
\drawedge[ELside=l,ELpos=48, ELdist=.5](q1p,q2p){{\small $1-\mygreen{x_1}$}}


\drawedge[ELside=r,ELpos=46, ELdist=.5](q2,q3){$t_2$}
\drawedge[ELside=l,ELpos=54, ELdist=.5](q2,q4){$1-t_2$}

\drawedge[ELside=r,ELpos=46, ELdist=.5](q2p,q3p){$t_3$}
\drawedge[ELside=l,ELpos=54, ELdist=.5](q2p,q4p){$1-t_3$}


\drawedge[ELside=r,ELpos=48, ELdist=.5](q3,l3){\myblue{{\small $x_2$}}}
\drawedge[ELside=l,ELpos=51, ELdist=.5](q3,l4){{\small $1\!-\!\myblue{x_2}$}}
\drawedge[ELside=r,ELpos=48, ELdist=.5](q4,l5){\myred{{\small $x_3$}}}
\drawedge[ELside=l,ELpos=51, ELdist=.5](q4,l6){{\small $1-\myred{x_3}$}}

\drawedge[ELside=r,ELpos=48, ELdist=.5](q3p,l3p){\myred{{\small $x_3$}}}
\drawedge[ELside=l,ELpos=51, ELdist=.5](q3p,l4p){{\small $1-\myred{x_3}$}}
\drawedge[ELside=r,ELpos=48, ELdist=.5](q4p,l5p){\myblue{{\small $x_2$}}}
\drawedge[ELside=l,ELpos=51, ELdist=.5](q4p,l6p){{\small $1\!-\!\myblue{x_2}$}}




\end{gpicture}
}

%% file: figures/game-cex-mapped.tex
\scalebox{.64}{
\begin{gpicture}(100,86)(0,0)


\gasset{Nw=6,Nh=6,Nmr=3}


\node[Nmarks=n, Nmr=0](q0)(35,82){$p_0$}

\node[Nmarks=n](q1)(10,66){$p_1$}
\node[Nmarks=n](q1p)(60,66){$p'_1$}

\node[Nmarks=n](l1)(4,54){$1$}
\node[Nmarks=n](l2)(54,54){$2$}

\node[Nmarks=n](l3)(7.5,4){$3$}
\node[Nmarks=n](l4)(17.5,4){$4$}
\node[Nmarks=n](l5)(32.5,4){$5$}
\node[Nmarks=n](l6)(42.5,4){$6$}
\node[Nmarks=n](l3p)(57.5,4){$3$}
\node[Nmarks=n](l4p)(67.5,4){$4$}
\node[Nmarks=n](l5p)(82.5,4){$5$}
\node[Nmarks=n](l6p)(92.5,4){$6$}

\node[Nmarks=n, Nmr=0](q2)(25,44){$p_2$}
\node[Nmarks=n, Nmr=0](q2p)(75,44){$p'_2$}

\node[Nmarks=n](q3)(12.5,24){$p_3$}
\node[Nmarks=n](q3p)(62.5,24){$p_3$}
\node[Nmarks=n](q4)(37.5,24){$p_4$}
\node[Nmarks=n](q4p)(88.5,24){$p_4$}

\node[Nmarks=n, Nh=8, Nw=58, Nmr=2, dash={0.5 1}0](box)(35,66){}

\node[Nmarks=n, Nh=8, Nw=58, Nmr=2,dash={0.5 1}0](box)(50,44){}

\node[Nmarks=n, Nh=8, Nw=83, Nmr=2,dash={0.5 1}0](box)(50,24){}


\drawedge[ELside=r,ELpos=50, ELdist=.5](q0,q1){$z_1$}
\drawedge[ELside=l,ELpos=50, ELdist=.5](q0,q1p){$1-z_1$}


\drawedge[ELside=r,ELpos=63, ELdist=.5](q1,l1){\mygreen{$y_1$}}
\drawedge[ELside=r,ELpos=63, ELdist=.5](q1p,l2){\mygreen{$y_1$}}

\drawedge[ELside=l,ELpos=48, ELdist=.5](q1,q2){$1-\mygreen{y_1}$}
\drawedge[ELside=l,ELpos=48, ELdist=.5](q1p,q2p){ $1-\mygreen{y_1}$}


\drawedge[ELside=r,ELpos=46, ELdist=.5](q2,q3){$z_2$}
\drawedge[ELside=l,ELpos=54, ELdist=.5](q2,q4){$1-z_2$}

\drawedge[ELside=r,ELpos=46, ELdist=.5](q2p,q3p){$z_3$}
\drawedge[ELside=l,ELpos=54, ELdist=.5](q2p,q4p){$1-z_3$}


\drawedge[ELside=r,ELpos=48, ELdist=.5](q3,l3){\myblue{{\small $y_2$}}}
\drawedge[ELside=l,ELpos=51, ELdist=.5](q3,l4){{\small $1\!-\!\myblue{y_2}$}}
\drawedge[ELside=r,ELpos=48, ELdist=.5](q4,l5){\myblue{{\small $y_2$}}}
\drawedge[ELside=l,ELpos=51, ELdist=.5](q4,l6){{\small $1-\myblue{y_2}$}}

\drawedge[ELside=r,ELpos=48, ELdist=.5](q3p,l3p){\myblue{{\small $y_2$}}}
\drawedge[ELside=l,ELpos=51, ELdist=.5](q3p,l4p){{\small $1-\myblue{y_2}$}}
\drawedge[ELside=r,ELpos=48, ELdist=.5](q4p,l5p){\myblue{{\small $y_2$}}}
\drawedge[ELside=l,ELpos=51, ELdist=.5](q4p,l6p){{\small $1\!-\!\myblue{y_2}$}}




\end{gpicture}}

%% file: figures/game-cex-instantiated.tex
\scalebox{.64}{
\begin{gpicture}(100,86)(0,0)


\gasset{Nw=6,Nh=6,Nmr=3}


\node[Nmarks=n, Nmr=0](q0)(35,82){$q_0$}

\node[Nmarks=n](q1)(10,66){$q_1$}
\node[Nmarks=n](q1p)(60,66){$q'_1$}

\node[Nmarks=n](l1)(4,54){$1$}
\node[Nmarks=n](l2)(54,54){$2$}

\node[Nmarks=n](l3)(7.5,4){$3$}
\node[Nmarks=n](l6)(42.5,4){$6$}
\node[Nmarks=n](l4p)(67.5,4){$4$}
\node[Nmarks=n](l5p)(82.5,4){$5$}

\node[Nmarks=n, Nmr=0](q2)(25,44){$q_2$}
\node[Nmarks=n, Nmr=0](q2p)(75,44){$q'_2$}

\node[Nmarks=n](q3)(12.5,24){$q_3$}
\node[Nmarks=n](q3p)(62.5,24){$q'_3$}
\node[Nmarks=n](q4)(37.5,24){$q_4$}
\node[Nmarks=n](q4p)(88.5,24){$q'_4$}

\node[Nmarks=n, Nh=8, Nw=58, Nmr=2, dash={0.5 1}0](box)(35,66){}

\node[Nmarks=n, Nh=8, Nw=58, Nmr=2,dash={0.5 1}0](box)(50,44){}

\node[Nmarks=n, Nh=8, Nw=33, Nmr=2,dash={0.5 1}0](box)(50,24){}
\drawline[arcradius=2,AHnb=0,dash={0.5 1}0](18.5,32)(8.5,32)(8.5,20)(16.5,20)(16.5,30)(84.5,30)(84.5,20)(92.5,20)(92.5,32)(18.5,32)


\drawedge[ELside=r,ELpos=50, ELdist=.5](q0,q1){$t_1$}
\drawedge[ELside=l,ELpos=50, ELdist=.5](q0,q1p){$1-t_1$}


\drawedge[ELside=r,ELpos=50, ELdist=.5](q1,l1){$\frac{1}{2}$}
\drawedge[ELside=r,ELpos=63, ELdist=.5](q1p,l2){$\frac{1}{2}$}

\drawedge[ELside=l,ELpos=48, ELdist=.5](q1,q2){$\frac{1}{2}$}
\drawedge[ELside=l,ELpos=48, ELdist=.5](q1p,q2p){$\frac{1}{2}$}


\drawedge[ELside=r,ELpos=46, ELdist=.5](q2,q3){$t_2$}
\drawedge[ELside=l,ELpos=54, ELdist=.5](q2,q4){$1-t_2$}

\drawedge[ELside=r,ELpos=46, ELdist=.5](q2p,q3p){$t_3$}
\drawedge[ELside=l,ELpos=54, ELdist=.5](q2p,q4p){$1-t_3$}


\drawedge[ELside=r,ELpos=48, ELdist=.5](q3,l3){$1$}
\drawedge[ELside=l,ELpos=51, ELdist=.5](q4,l6){$1$}

\drawedge[ELside=l,ELpos=51, ELdist=.5](q3p,l4p){$1$}
\drawedge[ELside=r,ELpos=48, ELdist=.5](q4p,l5p){$1$}




\end{gpicture}}

%% file: figures/game-cex-mapped-instantiated.tex
\scalebox{.64}{
\begin{gpicture}(100,86)(0,0)

\put(100,86){\makebox(0,0)[l]{\tiny{$\cdot$}}}


\gasset{Nw=6,Nh=6,Nmr=3}


\node[Nmarks=n, Nmr=0](q0)(35,82){$p_0$}

\node[Nmarks=n](q1)(10,66){$p_1$}

\node[Nmarks=n](l1)(4,54){$1$}

\node[Nmarks=n](l3)(7.5,4){$3$}
\node[Nmarks=n](l4)(17.5,4){$4$}
\node[Nmarks=n](l5)(32.5,4){$5$}
\node[Nmarks=n](l6)(42.5,4){$6$}

\node[Nmarks=n, Nmr=0](q2)(25,44){$p_2$}

\node[Nmarks=n](q3)(12.5,24){$p_3$}
\node[Nmarks=n](q4)(37.5,24){$p_4$}



\node[Nmarks=n, Nh=8, Nw=34, Nmr=2,dash={0.5 1}0](box)(25,24){}


\drawedge[ELside=r,ELpos=50, ELdist=.5](q0,q1){$1$}


\drawedge[ELside=r,ELpos=50, ELdist=.5](q1,l1){$\frac{1}{2}$}

\drawedge[ELside=l,ELpos=48, ELdist=.5](q1,q2){$\frac{1}{2}$}


\drawedge[ELside=r,ELpos=53, ELdist=.5](q2,q3){$1-\myblue{y_2}$}
\drawedge[ELside=l,ELpos=47, ELdist=.5](q2,q4){\myblue{$y_2$}}



\drawedge[ELside=r,ELpos=48, ELdist=.5](q3,l3){\myblue{{\small $y_2$}}}
\drawedge[ELside=l,ELpos=51, ELdist=.5](q3,l4){{\small $1\!-\!\myblue{y_2}$}}
\drawedge[ELside=r,ELpos=48, ELdist=.5](q4,l5){\myblue{{\small $y_2$}}}
\drawedge[ELside=l,ELpos=51, ELdist=.5](q4,l6){{\small $1-\myblue{y_2}$}}





\end{gpicture}
}